\documentclass[12pt,draftclsnofoot,onecolumn]{IEEEtran}
\usepackage{amsfonts}
\usepackage{cite,graphicx,amsmath,amsthm}
\usepackage{subfigure}
\usepackage{fancyhdr}
\usepackage{dsfont}
\usepackage{array,color}
\usepackage{bm}
\usepackage{float}
\usepackage{algorithm}
\usepackage{algpseudocode}
\usepackage{multirow}
\usepackage{bbm}
\usepackage{graphicx}
\usepackage{float}
\usepackage{subfigure}
\usepackage{ragged2e}
\usepackage{booktabs}
\usepackage{bbding}
\usepackage{multirow}
\usepackage{makecell}
\usepackage{amsmath,amssymb,amsfonts}
\usepackage[]{caption2}
\usepackage[table,xcdraw]{xcolor}
\usepackage{cite}
\usepackage{geometry}

\allowdisplaybreaks[4]

\newtheorem{theorem}{Theorem}

\newtheorem{lemma}{Lemma}

\newtheorem{proposition}{Proposition}

\newtheorem{corollary}{Corollary}

\newtheorem{property}{Property}

\newtheorem{remark}{Remark}

\newtheorem{claim}{Claim}

\newtheorem{assumption}{Assumption}

\begin{document}
\title{{Unit-Modulus Wireless Federated Learning\\Via Penalty Alternating Minimization}}

\author{
\small
Shuai~Wang$^{*,\diamond}$, Dachuan Li$^{\diamond}$, Rui~Wang$^*$, Qi Hao$^{\diamond}$, Yik-Chung~Wu$^{\dag}$, and Derrick~Wing~Kwan~Ng$^{\star}$\\
$^*$Department of Electrical and Electronic Engineering, Southern University of Science and Technology, China\\
$^\diamond$Department of Computer Science and Engineering, Southern University of Science and Technology, China\\
$^\dag$Department of Electrical and Electronic Engineering, University of Hong Kong, Hong Kong\\
$^\star$School of Electrical Engineering and Telecommunications, University of New South Wales, Australia\\
E-mail: wangs3@sustech.edu.cn, dachuanli86@gmail.com, wang.r@sustech.edu.cn,\\hao.q@sustech.edu.cn, ycwu@eee.hku.hk, w.k.ng@unsw.edu.au
\thanks{
This paper has been accepted for publication in IEEE Global Communications Conference 2021.
}
}

\maketitle

\vspace{-0.2in}
\begin{abstract}
Wireless federated learning (FL) is an emerging machine learning paradigm that trains a global parametric model from distributed datasets via wireless communications.
This paper proposes a unit-modulus wireless FL (UMWFL) framework, which simultaneously uploads local model parameters and computes global model parameters via optimized phase shifting.
The proposed framework avoids sophisticated baseband signal processing, leading to both low communication delays and implementation costs.
A training loss bound is derived and a penalty alternating minimization (PAM) algorithm is proposed to minimize the nonconvex nonsmooth loss bound.
Experimental results in the Car Learning to Act (CARLA) platform show that the proposed UMWFL framework with PAM algorithm achieves smaller training losses and testing errors than those of the benchmark scheme.
\end{abstract}

\IEEEpeerreviewmaketitle

\section{Introduction}

Federated learning (FL) is a promising technique to reduce the communication overhead while protecting the data privacy at users for effective machine learning \cite{fed1}.
Although FL was originally developed for wire-line connected systems \cite{fed1}, the development of mobile intelligent systems such as autonomous driving \cite{zijian} calls for wireless connections between the server and users, giving rise to a new paradigm termed wireless FL or edge FL \cite{fed2,fed3,fed4,air2,air3,air4,air5}.
However, the convergence of wireless FL may require an exceedingly long time due to limited capacity of wireless channels during the uplink model aggregation step.
To reduce the transmission delay, various wireless FL designs have been proposed, which are mainly categorized into digital modulation \cite{fed2,fed3,fed4} and analog modulation \cite{air2,air3,air4,air5} methods.

For digital modulation systems, data from different users are multiplexed either in the time or the frequency domain \cite{massive}.
Current works on delay reduction focus on reducing 1) the number of model aggregation iterations \cite{fed2}, 2) the number of users \cite{fed3}, or 3) the number of bits for representing the gradient of back propagation in each iteration \cite{fed4}.
However, since these strategies involve various approximations of the FL procedure, the performance of learning would be degraded inevitably.
On the other hand, the key advantage of analog modulation \cite{air2,air3,air4,air5} over digital modulation arises from the ground-breaking idea of over-the-air computation (AirComp).
Specifically, if multiple users upload their local parameters simultaneously, a superimposed signal, which represents a weighted sum of individual model parameters, is observed at the edge server.
By performing the minimum mean square error (MMSE) detection on the superimposed signal, an estimate of the global parameter vector can be obtained, thereby significantly shortening the required transmission time.
However, due to channel fading and noise in wireless systems, AirComp employed in single-antenna systems \cite{air2,air3} could result in large error in the estimation of global model parameters at the edge server, leading to slow convergence of FL iterations.
As a remedy, adopting MIMO beamforming \cite{air4,air5} could reduce the parameter transmission error by aligning the beams carrying the local parameters' information to the same desired spatial direction.
Nonetheless, the current transmit and receive beamforming designs in MIMO AirComp systems involve exceedingly high radio frequency (RF) chain costs and high computational complexities \cite{air4,air5}, preventing their practical implementation.

To fill the research gap, this paper proposes the unit-modulus wireless FL (UMWFL) framework which integrates AirComp and unit-modulus analog beamforming.
Specifically, the edge users possess a number of training data for local model updates.
Then, the trained model parameters are uploaded to the server via analog modulation.
To reduce the implementation costs, the edge server does not process the received model parameters at the baseband.
Instead, it applies a phase shift network in the RF domain and an RF sampling module to connect received antennas and transmit antennas for global model updates.
Upon receiving the broadcast, all users feed the received signals to an analog demodulator for parameter extraction.
Note that our UMWFL framework can significantly reduce the required RF chains in MIMO FL systems, thereby reducing the associated hardware and energy costs.
The proposed UMWFL framework also alleviates the straggler effects in both model uploading and broadcasting procedures.
On the other hand, despite the UMWFL problem being highly nonconvex, a large-scale optimization algorithm, termed penalty alternating minimization (PAM), is developed.
Experimental results show that the learning performance of the proposed PAM significantly outperforms the benchmark scheme without phase optimization.

\emph{Notation}: Italic letters, lowercase and uppercase bold letters represent scalars, vectors, and matrices, respectively.
Curlicue letters stand for sets and $|\cdot|$ is the cardinality of a set.
The operators $\|\cdot\|_2,(\cdot)^{T},(\cdot)^{H},(\cdot)^{-1}$ are the $\ell_2$-norm, transpose, Hermitian, and inverse of a matrix, respectively.
The operators $\partial f$ and $\nabla f$ are the partial derivative and the gradient of the function $f$.
The operators $\mathrm{vec}$ and $\mathrm{mat}$ denote the vectorization of a matrix and the matricization of a vector.
The function $[x]^+=\mathrm{max}(x,0)$, $\mathrm{Re}(x)$ takes the real part of $x$, $\mathrm{Im}(x)$ takes the imaginary part of $x$, $\mathrm{conj}(x)$ takes the conjugate of $x$, and $|x|$ is the modulus of $x$.
$\mathbf{I}_{N}$ denotes the $N\times N$ identity matrix and $\otimes$ denotes the Kronecker product.
Finally, $\mathrm{j}=\sqrt{-1}$, $\mathbb{E}(\cdot)$ denotes the expectation of a random variable and $\mathcal{O}(\cdot)$ is the big-O notation standing for the order of arithmetic operations.

\begin{figure}[!t]
 \centering
\includegraphics[width=0.85\textwidth]{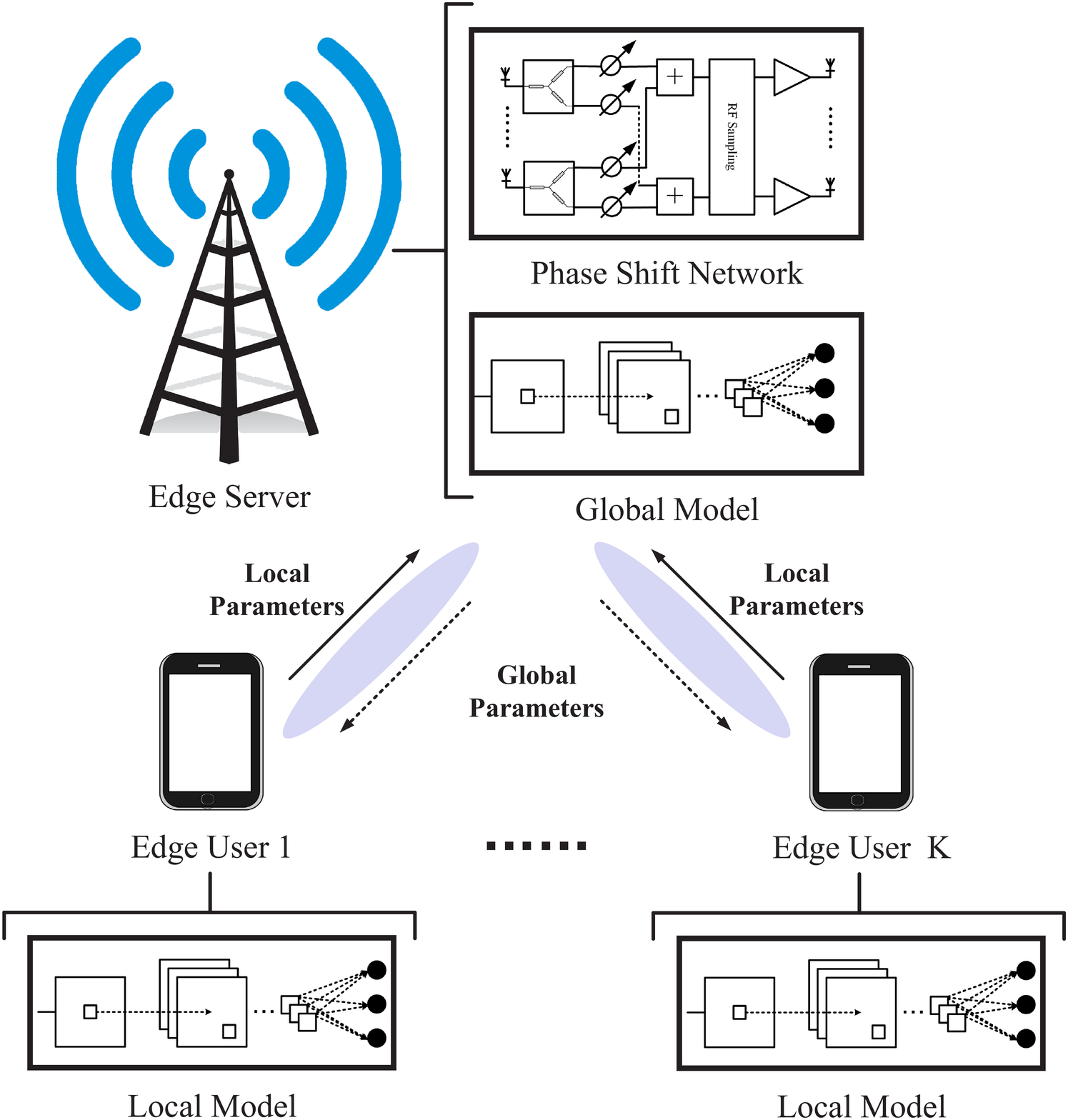}
  \caption{The UMWFL system with a phase shifting edge server.}
\end{figure}

\section{Unit Modulus Wireless Federated Learning}

\setcounter{secnumdepth}{4}
We consider a UMWFL system shown in Fig.~1, which consists of an edge server equipped with $N$ antennas and $K$ single-antenna mobile users. The dataset and model parameter vector at user $k$ are denoted as $\mathcal{D}_k$ and $\mathbf{x}_k\in\mathbb{R}^{M\times 1}$, respectively.
Mathematically, the UMWFL procedure aims to solve the following optimization problem:
\begin{align}\label{FL}
\mathop{\mathrm{min}}_{\substack{\{\mathbf{x}_k\},\bm{\theta}}}
\quad&\underbrace{\frac{1}{\sum_{k=1}^{K}|\mathcal{D}_k|}\sum_{k=1}^K\sum_{\mathbf{d}_{k,l}\in\mathcal{D}_k}\Theta(\mathbf{d}_{k,l}, \bm{\theta})}_{:=\Lambda(\bm{\theta})}
\nonumber\\
\mathrm{s.t.}\quad&\mathbf{x}_1=\cdots=\mathbf{x}_K= \bm{\theta},
\end{align}
where $\Theta(\mathbf{d}_{k,l}, \bm{\theta})$ is the loss function corresponding to a single sample $\mathbf{d}_{k,l}$ ($1\leq l \leq |\mathcal{D}_k|$) in $\mathcal{D}_k$ given parameter vector $\bm{\theta}$, while $\Lambda(\bm{\theta})$ denotes the global loss function to be minimized.
The training of UMWFL model parameters (i.e., solving \eqref{FL}) in the considered edge system is naturally a distributed and iterative procedure, where each iteration involves four steps:
1) updating the local parameter vectors $(\mathbf{x}_1,\cdots,\mathbf{x}_K)$ using $\{\mathcal{D}_1,\cdots,\mathcal{D}_K\}$ at users $(1,\cdots,K)$, respectively;
2) transforming the local parameters $(\mathbf{x}_1,\cdots,\mathbf{x}_K)$ into transmit symbols $(\mathbf{s}_1,\cdots,\mathbf{s}_K)$ via power and phase designs;
3) aggregating $(\mathbf{x}_1,\cdots,\mathbf{x}_k)$ in an analog manner at the edge server;
and 4) broadcasting the results to users.
The details for the $i$-th iteration of UMWFL are elaborated below.

\subsubsection{Local Updating}
In the first step, let $\mathbf{x}^{[i]}_k(0)\in\mathbb{R}^{M\times 1}$ be the local parameter vector at user $k$ at the beginning of the $i$-th iteration.
To update $\mathbf{x}^{[i]}_k(0)$, user $k$ minimizes the loss function $\frac{1}{|\mathcal{D}_k|}\sum_{\mathbf{d}_{k,l}\in\mathcal{D}_k}\Theta(\mathbf{d}_{k,l}, \mathbf{x}_k)$ via the gradient descent method as
\begin{align}\label{sgd}
&\mathbf{x}^{[i]}_k(\tau+1)=\mathbf{x}^{[i]}_k(\tau)-\frac{\varepsilon}{|\mathcal{D}_k|}\sum_{\mathbf{d}_{k,l}\in\mathcal{D}_k}\nabla_{\mathbf{x}}\Theta\left(\mathbf{d}_{k,l},\mathbf{x}^{[i]}_k(\tau)\right),
\end{align}
where $\varepsilon>0$ is the step-size and $\tau$ is from $0$ to $E-1$ with $E$ being the number of local updates.
Then, $\{\mathbf{x}^{[i]}_k(E)|\forall k\}$ from all users are uploaded to the edge server.

\subsubsection{Parameter Uploading}
In the second step, user $k$ encodes its local parameter vector $\mathbf{x}^{[i]}_k(E)$ into a complex vector $\mathbf{s}^{[i]}_k$.
Since the model parameters in deep learning are real-valued numbers, in order to reduce the transmission time, every two model parameters are transmitted as a complex number.  That is,
\begin{align}
\mathbf{s}_k^{[i]}=&
\frac{\overbrace{\sqrt{p^{[i]}_k}\,\mathrm{e}^{\rm{j}\phi_k^{[i]}}}^{:=t_k^{[i]}}}{\sqrt{2\eta^{[i]}}}\,\Big[x_{k,1}^{[i]}(E)+\mathrm{j}\,x_{k,2}^{[i]}(E),\cdots,x_{k,M-1}^{[i]}(E)+\mathrm{j}\,x_{k,M}^{[i]}(E)\Big]^T, \label{sk}
\end{align}
where $p_k^{[i]}\in\mathbb{R}$ and $\phi_k^{[i]}$ are the transmit power and phase at user $k$, $\mathbf{s}_k^{[i]}\in\mathbb{C}^{S\times 1}$ with $S=M/2$, $x_{k,m}^{[i]}(E)$ is the $m$-th element of $\mathbf{x}_k^{[i]}(E)$.
The scaling factor $\eta^{[i]}$ is $\eta^{[i]}=\frac{1}{K}\sum_{k=1}^K\overline{\eta}_k^{[i]}$ with $\overline{\eta}_k^{[i]}=\frac{1}{M}\|\mathbf{x}_k^{[i]}(E)\|_2^2$ such that the average power of $\mathbf{s}_k^{[i]}$ is $\frac{1}{S}\mathbb{E}[\|\mathbf{s}_k^{[i]}\|_2^2]=p_k^{[i]}$.
To facilitate the subsequent derivations, we define the transmit coefficient $\{t_k^{[i]}\}$ in \eqref{sk} and $\{p_k^{[i]},\varphi_k^{[i]}\}$ can be recovered from $\{t_k^{[i]}\}$.

\subsubsection{Parameter Aggregation}
In the third step, the received signal $\mathbf{R}^{[i]}\in\mathbb{C}^{N\times S}$ at the server is
\begin{align}\label{Ri}
&\mathbf{R}^{[i]}=\sum_{k=1}^K\mathbf{h}_k^{[i]}(\mathbf{s}_{k}^{[i]})^T+\mathbf{Z}^{[i]},
\end{align}
where $\mathbf{h}_k^{[i]}\in\mathbb{C}^{N\times 1}$ is the uplink channel vector from user $k$ to the server and $\mathbf{Z}^{[i]}\in\mathbb{C}^{N\times S}$ is the matrix of the additive white Gaussian noise with covariance matrix
$\mathbb{E}\left[\mathrm{vec}(\mathbf{Z}^{[i]})\mathrm{vec}(\mathbf{Z}^{[i]})^H\right]=\sigma^2_b\mathbf{I}_{NS}$, where $\sigma^2_b$ is the noise power at the server.
Upon receiving the superimposed signal, the server processes $\mathbf{R}^{[i]}$ using a phase shift network $\mathbf{F}^{[i]}\in\mathbb{C}^{N\times N}$.
The phase shifted signals are then fed to the RF sampling module \cite{rfsampling} for data caching.
As shown in Fig.~1, the matrix $\mathbf{F}^{[i]}$ requires unit-modulus constraints on all of its elements \cite{analog,analog2}, i.e., $\mathbf{F}^{[i]}\in\mathcal{F}$ where
\begin{align}
&\mathcal{F}=\{\mathbf{F}:|F_{l,l'}^{[i]}|=1,\quad\forall l,l'\}. \label{setF}
\end{align}

\subsubsection{Parameter Decoding}
Finally, the server broadcasts $\sqrt{\gamma}\,\mathbf{F}^{[i]}\mathbf{R}^{[i]}\in\mathbb{C}^{N\times S}$ to all the users, where $\gamma>0$ is the power scaling factor adopted at the edge server.
The received signal at user $k$ is
\begin{align}
&\left(\mathbf{y}^{[i]}_k\right)^T=\left(\mathbf{g}_k^{[i]}\right)^H
\sqrt{\gamma}\,\mathbf{F}^{[i]}\mathbf{R}^{[i]}+\left(\mathbf{n}^{[i]}_k\right)^T, \label{yk}
\end{align}
where $\mathbf{g}_k^{[i]}\in\mathbb{C}^{N\times 1}$ is the downlink channel vector from the server to user $k$ and $\mathbf{n}_k^{[i]}\in\mathbb{C}^{S\times 1}$ is the vector of the additive white Gaussian noise with covariance matrix $\sigma^2_k\mathbf{I}_S$, where $\sigma^2_k$ is the noise power at user $k$.
User $k$ applies a normalization coefficient $r^{[i]}_k\in\mathbb{C}$ to $\mathbf{y}^{[i]}_k$ as
\begin{align}
\mathbf{x}^{[i+1]}_k(0)
=&
\sqrt{2\eta^{[i]}}\,\Big[\mathrm{Re}(r_k^{[i]}y^{[i]}_{k,1}),\mathrm{Im}(r^{[i]}_ky^{[i]}_{k,1}),\cdots,\mathrm{Re}(r^{[i]}_ky^{[i]}_{k,L}),\mathrm{Im}(r^{[i]}_ky^{[i]}_{k,L})\Big]^T, \label{nextround}
\end{align}
where $y^{[i]}_{k,l}$ is the $l$-th element of $\mathbf{y}^{[i]}_k$.
It can be seen that $\mathbf{x}^{[i+1]}_k(0)$ is the starting point for the $(i+1)$-th iteration.
This completes one UMWFL round.

\section{Training Loss Minimization}

Ideally, the optimization of $\{\mathbf{F}^{[i]},r_k^{[i]},t_k^{[i]}\}$ should be performed to minimize the training loss, i.e., $\min~\Lambda(\mathbf{x}_k^{[i+1]}(0))$.
However, the analytical expression of $\mathbb{E}\left[\Lambda(\mathbf{x}_k^{[i+1]}(0))\right]$, where the expectation is taken over receiver noises and model parameters, is usually challenging to derive.
As a compromise approach, we resort the use of the upper bound of the expectation of $\mathbb{E}\left[\Lambda(\mathbf{x}_k^{[i+1]}(0))-\Lambda(\bm{\theta}^*)\right]$ as a metric to capture the degradation on training loss.
The upper bound depends on the MSE of the model parameters' estimation $\mathbb{E}\left[\|\mathbf{x}^{[i+1]}_k(0)-\bm{\theta}^{[i]}\|_2^2\right]$ where
\begin{align}
&\bm{\theta}^{[i]}=\sum_{k=1}^K\frac{|\mathcal{D}_k|}{\sum_{l=1}^{K}|\mathcal{D}_l|}\mathbf{x}_k^{[i]}(E) \label{global}
\end{align}
is equivalent to the gradient descent of the objective function of \eqref{FL}.
Specifically, based on \eqref{nextround} and \eqref{global}, the MSE between the received local parameter $\mathbf{x}^{[i+1]}_k(0)$ and the target local parameter $\bm{\theta}^{[i]}$ at the $i$-th UMWFL iteration is
\begin{align}\label{mse}
&\mathbb{MSE}_k^{[i]}\left(\mathbf{F}^{[i]}, r_k^{[i]},\{t_k^{[i]}\} \right)
=\mathbb{E}\left[\Big\|\mathbf{x}^{[i+1]}_k(0)-\bm{\theta}^{[i]}\Big\|_2^2\right]
\nonumber\\
&=2\eta^{[i]} S\Bigg[\gamma\sum_{j=1}^K\Big|r_k^{[i]}(\mathbf{g}_k^{[i]})^H\mathbf{F}^{[i]}\mathbf{h}_j^{[i]}\sqrt{t_j^{[i]}}-\alpha_j\Big|^2
+\gamma\sigma^2_b\|r_k^{[i]}(\mathbf{g}_k^{[i]})^H\mathbf{F}\|_2^2+\sigma_k^2|r_k^{[i]}|^2\Bigg],
\end{align}
where $\alpha_k=\frac{|\mathcal{D}_k|}{\sum_{l=1}^{K}|\mathcal{D}_l|}$ and the equality is due to \eqref{sk}, \eqref{Ri}, \eqref{yk}, \eqref{nextround}, \eqref{global}, and the independence among $\{\mathbf{s}_k|\forall k\}$.
Having obtained the MSE expression, the next step is to derive the relationship between the loss bound and the MSE.
To this end, we first introduce the following assumption imposed on the loss function.
\begin{assumption}
(i) The function $\Lambda(\mathbf{x})$ is $\mu$-strongly convex.
(ii) The function $\sum_{\mathbf{d}_{k,l}\in\mathcal{D}_k}\Theta(\mathbf{d}_{k,l}, \mathbf{x})$ is twice differentiable and satisfies $\sum_{\mathbf{d}_{k,l}\in\mathcal{D}_k}\nabla^2_{\mathbf{x}}\Theta(\mathbf{d}_{k,l}, \mathbf{x})\preceq L\mathbf{I}$.
\end{assumption}

Under Assumption 1, the relationship between $\Lambda(\mathbf{x}_k^{[i+1]}(0))$ and $\Lambda(\bm{\theta}^{*})$ is summarized in the following theorem.
\begin{theorem}
With $(\varepsilon,E)=(\frac{\sum_{k=1}^{K}|\mathcal{D}_k|}{KL}, 1)$, the UMWFL scheme satisfies
\begin{align}
&\mathbb{E}\left[\Lambda(\mathbf{x}_k^{[i+1]}(0))-\Lambda(\bm{\theta}^*)\right]
\leq
\sum_{i'=0}^{i}A^{[i']}\,\mathop{\mathrm{max}}_{k=1,\cdots,K}\mathbb{MSE}^{[i']}_k,
\end{align}
for any $\{\mathbf{F}^{[i']},r_k^{[i']},t_k^{[i']}\}_{i'=0}^i$ as $i\rightarrow+\infty$, where
\begin{align}
&A^{[i']}=
\frac{KL\left(3+2K^{-1}\right)}{2\sum_{k=1}^{K}|\mathcal{D}_k|}
\left(1-\frac{\mu\sum_{k=1}^{K}|\mathcal{D}_k|}{KL}\right)^{i-i'}.
\end{align}
\end{theorem}
\begin{proof}
We first derive the upper bound between consecutive parameter vectors.
Then the upper bound between consecutive losses can be derived.
Finally, we derive the upper bound of the sequence of loss function values.
Please refer to \cite{umair} for detailed derivations.
\end{proof}

Theorem 1 shows a diminishing $A^{[i']}\rightarrow 0$ for a large $i-i'$, meaning that the impact from earlier UMWFL iterations vanishes as the UMWFL continues.
On the other hand, if $\mathbb{MSE}^{[i']}_k\rightarrow 0$ for all $k$, then $\Lambda(\mathbf{x}_k^{[i+1]})$ is an unbiased estimate of $\Lambda(\bm{\theta}^*)$.
This demonstrates the effectiveness of UMWFL in the asymptotic regime.
The convexity and smoothness in Assumption 1 have been adopted in most loss bound analysis of FL (e.g., \cite{fed3}).
Although it seems to be restrictive for some realistic applications, analysis under Assumption 1 could provide important insights of the behavior of UMWFL in nonconvex cases.

Based on Theorem 1, the training loss minimization in UMWFL systems is formulated as
\begin{align}
\mathop{\mathrm{min}}_{\substack{\{\mathbf{F}^{[i']}\in\mathcal{F},r_k^{[i']},|t_k^{[i]}|^2\leq P_0\}}}
\quad&
\sum_{i'=0}^{i}A^{[i']}\,\mathop{\mathrm{max}}_{k=1,\cdots,K}\mathbb{MSE}^{[i']}_k.
\end{align}
It can be seen that the above problem and constraints can be decoupled for each iteration and the minimization at the $i$-th UMWFL iteration, $\forall i$, is given by
\begin{subequations}
\begin{align}
\mathcal{P}:\mathop{\mathrm{min}}_{\substack{\mathbf{F},\,\{r_k,t_k\}}}
\quad&
\mathop{\mathrm{max}}_{k=1,\cdots,K}~\gamma\sum_{j=1}^K\Big|r_k\mathbf{g}_k^H\mathbf{F}\mathbf{h}_jt_j-\alpha_j\Big|^2
+\gamma\sigma^2_b\|r_k(\mathbf{g}_k)^H\mathbf{F}\|_2^2+\sigma_k^2|r_k|^2
\label{P0}
\\
\quad\quad\quad\mathrm{s. t.}\quad\quad
&\mathbf{F}\in\mathcal{F},\quad |t_k|^2\leq P_0, \quad \forall k,
\end{align}
\end{subequations}
where the UMWFL iteration index $i'$ is removed since there is no dependence among different UMWFL iterations.

Problem $\mathcal{P}$ is generally NP-hard due to the unit-modulus constraints \cite{analog,analog2}.
In addition, the coupling between variables $\{r_k,t_k\}$ and $\mathbf{F}$ introduces nonlinearity and nonconvexity to the considered problem.
Finally, the large dimension of $\mathbf{F}$ call for the design of low-complexity algorithms in the scenario with massive numbers of antennas.

\section{Penalty Alternating Minimization for UMWFL}

In this section, the PAM algorithm, which consists of two nested layers of iterations (i.e., an outer-layer iteration and an inner-layer iteration), will be proposed to optimize the system performance.
Below we first introduce the outer-layer iteration.

\subsection{Outer-Layer Iteration Via Alternating Optimization}

To resolve the coupling between variables $\{r_k,t_k\}$ and $\mathbf{F}$, this paper adopts an alternating optimization framework \cite{beck}, which optimizes one design variable at a time with others being fixed.
Starting with an initial solution $\{\mathbf{F}^{(0)},r_k^{(0)},t_k^{(0)}\}$, the entire procedure
solving problem $\mathcal{P}$ for the $(n+1)$-th outer iteration, $\forall n$, can be elaborated below:
\begin{subequations}
\begin{align}
\mathbf{F}^{(n+1)}=\mathop{\mathrm{arg~min}}_{\substack{\mathbf{F}}}
&
\mathop{\mathrm{max}}_{k}~\Bigg(\sum_{j=1}^K\Big|r_k^{(n)}\mathbf{g}_k^H\mathbf{F}\mathbf{h}_jt_j^{(n)}-\alpha_j\Big|^2
+\sigma^2_b\|r_k^{(n)}\mathbf{g}_k^H\mathbf{F}\|_2^2\Bigg) \nonumber\\
\quad\quad\mathrm{s. t.}\quad
&|F_{l,l'}|=1,\quad \forall l,l', \label{problemF}
\\
\{r_k^{(n+1)}\}=\mathop{\mathrm{arg~min}}_{\substack{\{r_k\}}}
&
\mathop{\mathrm{max}}_{k}
\Bigg(
\gamma\sigma^2_b\|r_k\mathbf{g}_k^H\mathbf{F}^{(n+1)}\|_2^2+\sigma_k^2|r_k|^2
+
\gamma\sum_{j=1}^K\Big|r_k\mathbf{g}_k^H\mathbf{F}^{(n+1)}\mathbf{h}_jt_j^{(n)}-\alpha_j\Big|^2
\Bigg), \label{problemq}
\\
\{t_k^{(n+1)}\}=\mathop{\mathrm{arg~min}}_{\substack{\{t_k\}}}
&
\mathop{\mathrm{max}}_{k}~\sum_{j=1}^K\Big|r_k^{(n+1)}\mathbf{g}_k^H\mathbf{F}^{(n+1)}\mathbf{h}_jt_j-\alpha_j\Big|^2 \nonumber\\
\quad\quad\mathrm{s. t.}\quad
&|t_k|^2\leq P_0,\quad k=1,\cdots,K, \label{problemp}
\end{align}
\end{subequations}
where $\{\mathbf{F}^{(n)},r_k^{(n)},t_k^{(n)}\}$ is the solution at the $n$-th outer iteration.
The iterative procedure stops until $n$ reaches the maximum iteration number $n=N_{\mathrm{max}}$.

It can be seen that problem \eqref{problemq} is a typical least squares problem.
The optimal solution is found by setting the derivative $\partial \mathbb{MSE}_k/\partial \mathrm{conj}(r_k)$ to zero, which yields
\begin{align}
r_k^{(n+1)}=&
\frac{\sum_{j=1}^K\alpha_j\mathrm{conj}\left(\mathbf{g}_k^H\mathbf{F}^{(n+1)}\mathbf{h}_jt_j^{(n)}\right)}
{\sum_{j=1}^K|t_j^{(n)}\mathbf{g}_k^H\mathbf{F}^{(n+1)}\mathbf{h}_j|^2
+\sigma^2_b\|\mathbf{g}_k^H\mathbf{F}^{(n+1)}\|_2^2
+\frac{\sigma_k^2}{\gamma}}. \label{qk}
\end{align}
On the other hand, problem \eqref{problemp} is a convex optimization problem, which can be solved with a complexity of $\mathcal{O}(K^{3.5})$ via CVX, a Matlab software package for solving convex problems based on interior point method (IPM).
Consequently, the key challenge is to tackle the nonconvex problem \eqref{problemF}.

While problem \eqref{problemF} can be transformed into a convex problem via semidefinite relaxation (SDR), the computational complexity for solving the SDR problem is at least $\mathcal{O}\left(N^{7}\right)$ \cite{opt1}.
Since $N$ is usually much larger than $K$, this method is not desirable.
In the following, a new algorithm termed PAM, which decomposes \eqref{problemF} into smaller subproblems that are solved by closed-form updates, is proposed for achieving both excellent performance and significantly lower computational complexities.

\subsection{Inner-Layer Iteration Via PAM}

Since $\mathbf{F}$ is a matrix, its vectorization is given as $\mathbf{f}=\mathrm{vec}\left(\mathbf{F}\right)\in\mathbb{C}^{N^2\times 1}$.
Applying $\mathrm{Tr}\left(\mathbf{A}\mathbf{X}\mathbf{B}\mathbf{X}^T\right)=\mathrm{vec}(\mathbf{X})^T\left(\mathbf{B}^T\otimes\mathbf{A}\right)\mathrm{vec}(\mathbf{X})$ \cite{vec},
we have
\begin{align}
r_k^{(n)}\mathbf{g}_k^H\mathbf{F}\mathbf{h}_jt_j^{(n)}&
=(\mathbf{a}_{k,j}^{(n)})^H\mathbf{f},
\\
\sigma^2_b\|r_k^{(n)}\mathbf{g}_k^H\mathbf{F}\|_2^2&
=\mathbf{f}^H\mathbf{G}_k^{(n)}\mathbf{f},
\end{align}
where
\begin{align}
\mathbf{G}_k^{(n)}&=\sigma^2_b|r_k^{(n)}|^2\mathbf{I}_N\otimes\left(\mathbf{g}_k\mathbf{g}_k^H\right),
\\
\mathbf{a}_{k,j}^{(n)}&=\left[r_k^{(n)}t_j^{(n)}\left(\mathbf{h}_j^T\otimes\mathbf{g}_k^H\right)\right]^H.
\end{align}

Problem \eqref{problemF} is thus re-formulated as
\begin{align}
\mathcal{P}_F:\mathop{\mathrm{min}}_{\substack{\mathbf{f}}}
\quad&\mathop{\mathrm{max}}_{k=1,\cdots,K}~\left(
\sum_{j=1}^K\Big|(\mathbf{a}_{k,j}^{(n)})^H\mathbf{f}-\alpha_j\Big|^2+
\mathbf{f}^H\mathbf{G}_k^{(n)}\mathbf{f}\right)\nonumber\\
\quad\quad\quad\mathrm{s. t.}\quad
&|f_{l}|=1,\quad l=1,\cdots,N^2. \label{PF}
\end{align}
To handle the nonseparable objective function, variable splitting of $\mathbf{f}$ is proposed such that $\mathbf{f}=\mathbf{u}_1=\cdots=\mathbf{u}_K$, where $\{\mathbf{u}_k\}$ are auxilliary variables.
Moreover, to handle the unit-modulus constraints, another auxilliary variable $\mathbf{z}=\mathbf{f}$ is introduced.
For the newly introduced equality constraints, they can be transformed into quadratic penalties in the objective function \cite{penalty}.
Thus, $\mathcal{P}_F$ is approximately transformed into
\begin{align}
\mathop{\mathrm{min}}_{\substack{\mathbf{f},\mathbf{z},\{\mathbf{u}_{k}\}}}
\quad&\mathop{\mathrm{max}}_{k=1,\cdots,K}~\left(
\sum_{j=1}^K\Big|(\mathbf{a}_{k,j}^{(n)})^H\mathbf{u}_{k}-\alpha_j\Big|^2+
\mathbf{u}_{k}^H\mathbf{G}_k^{(n)}\mathbf{u}_{k}\right)
+\rho\left(\frac{1}{K}\sum_{j=1}^K\|\mathbf{u}_{j}-\mathbf{f}\|_2^2+\|\mathbf{z}-\mathbf{f}\|_2^2\right)
\nonumber\\
\mathrm{s. t.}\quad
&|z_{l}|=1,\quad l=1,\cdots,N^2, \label{PFa}
\end{align}
where $\rho$ is a penalty parameter for fine-tuning.
It can be proved that $\mathcal{P}_F$ and \eqref{PFa} are equivalent problems as $\rho\rightarrow+\infty$ \cite{penalty}.
However, this case also leads to the gradient norm of the objective function of \eqref{PFa} being infinite, making \eqref{PFa} difficult to solve.
Therefore, $\rho$ controls the tradeoff between approximation error and difficulty in solving \eqref{PFa}.

We address \eqref{PFa} using alternating minimization, in which the cost function is iteratively minimized with respect to one variable whereas the others are fixed.
Starting with an initial $\mathbf{f}^{(0)}=\mathbf{z}^{(0)}=\mathbf{u}_{k}^{(0)}=\mathrm{vec}\left(\mathbf{F}^{(n)}\right)$, the whole process consists of iteratively solving
\begin{subequations}
\begin{align}
\mathbf{u}_{k}^{(m+1)}=\mathop{\mathrm{arg~min}}_{\substack{\mathbf{u}_{k}}}
\quad&\sum_{j=1}^K\Big|(\mathbf{a}_{k,j}^{(n)})^H\mathbf{u}_{k}-\alpha_j\Big|^2+
\mathbf{u}_{k}^H\mathbf{G}_k^{(n)}\mathbf{u}_{k}
+\frac{\rho}{K}\|\mathbf{u}_{k}-\mathbf{f}^{(m)}\|_2^2,\quad \forall k, \label{PFa1}
\\
\mathbf{f}^{(m+1)}=\mathop{\mathrm{arg~min}}_{\substack{\mathbf{f}}}
\quad&\rho\Bigg(\frac{1}{K}\sum_{j=1}^K\|\mathbf{u}_{j}^{(m+1)}-\mathbf{f}\|_2^2
+\|\mathbf{z}^{(m)}-\mathbf{f}\|_2^2\Bigg), \label{PFa2}
\\
\mathbf{z}^{(m+1)}=\mathop{\mathrm{arg~min}}_{\substack{|z_{l}|=1, \forall l}}
\quad&\rho\|\mathbf{z}-\mathbf{f}^{(m+1)}\|_2^2, \label{PFa3}
\end{align}
\end{subequations}
where $m$ is the inner iteration index.
It can be verified that the objective function of \eqref{PFa} is strongly convex.
Therefore, despite the non-differentiability of the objective, the alternating minimization \eqref{PFa1}--\eqref{PFa3} is guaranteed to converge to a stationary point of \eqref{PFa} \cite{beck}.
The iterative procedure stops until $m$ reaches the maximum iteration number $m=M_{\mathrm{max}}$.

The remaining question is how to solve \eqref{PFa1}--\eqref{PFa3} optimally.
We notice that problems \eqref{PFa1} and \eqref{PFa2} are standard least squares problems, thus their solutions are given by the following closed-form expressions
\begin{align}
\mathbf{u}_{k}^{(m+1)}&=\left(
\sum_{j=1}^K\mathbf{a}_{k,j}^{(n)}(\mathbf{a}_{k,j}^{(n)})^H+\mathbf{G}_k^{(n)}+\rho\mathbf{I}
\right)^{-1}\left(\sum_{j=1}^K\alpha_j\mathbf{a}_{k,j}^{(n)}+\frac{\rho}{K}\mathbf{f}^{(m)}\right), \label{ukm}
\\
\mathbf{f}^{(m+1)}&=\frac{1}{2}\left(\frac{1}{K}\sum_{j=1}^K\mathbf{u}_{j}^{(m+1)}+\mathbf{z}^{(m)}\right), \label{fm}
\end{align}
respectively.
On the other hand, problem \eqref{PFa3} is the projection of $\mathbf{f}^{(m+1)}$ onto unit-modulus constraint and the optimal solution is simply
\begin{align}
\mathbf{z}^{(m+1)}=\mathrm{exp}\left(\mathrm{j}\angle\mathbf{f}^{(m+1)}\right). \label{zm}
\end{align}

\begin{figure*}[!t]
\centering
\subfigure[]{\includegraphics[width=0.99\textwidth]{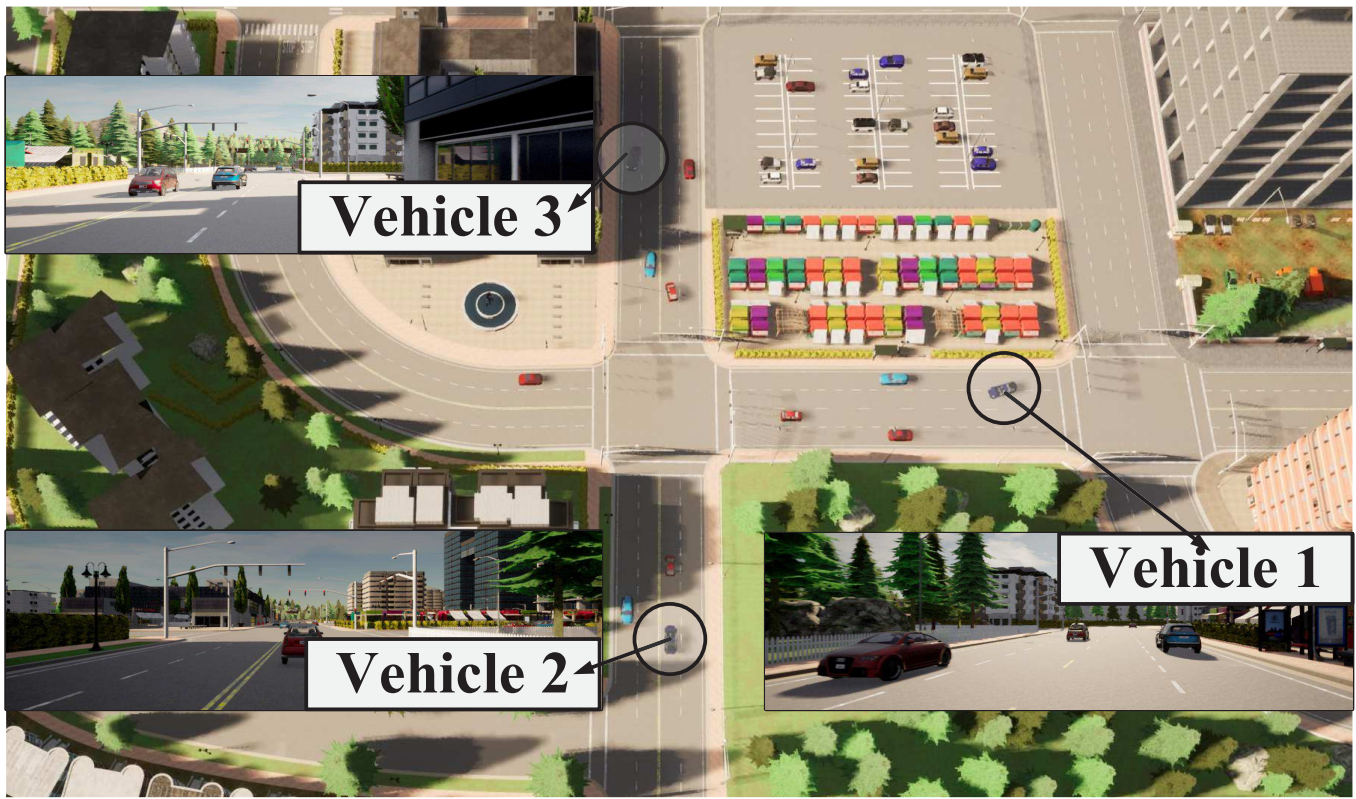}}
\subfigure[]{\includegraphics[height=0.42\textwidth]{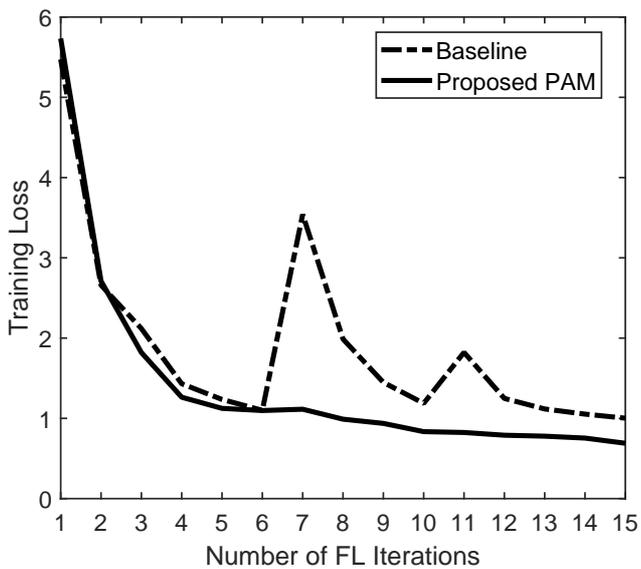}}
\subfigure[]{\includegraphics[height=0.42\textwidth]{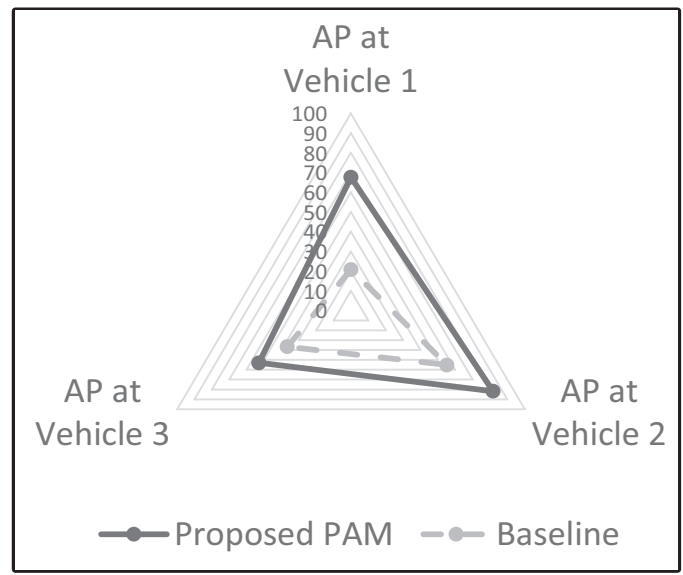}}
  \caption{Comparison between the proposed and benchmark schemes when $N=8$ with $K=3$: a) Multi-vehicle heterogeneous datasets in ``Town 05'' map; b) Training loss versus the number of UMWFL iterations; c) Average precision at $\textrm{IoU}=0.5$ on the testing datasets at different vehicles. The axis of each vehicle scales from average precision (AP) $0\%$ to $100\%$ with a step size of $10\%$.}
\end{figure*}

\subsection{Summary and Complexity Analysis of PAM}

In summary, the complete PAM algorithm for solving problem $\mathcal{P}$ is summarized in \textbf{Algorithm 1}.
The computational complexity (CP) is dominated by the steps at lines 4--6 in Algorithm 1, of which the CPs are given by $\mathcal{O}(KN^2)$, $\mathcal{O}(KN^2)$, and $\mathcal{O}(N^2)$, respectively.
Therefore, the total CP of PAM is $\mathcal{O}(N_{\rm{max}}M_{\rm{max}}KN^2)$

\begin{algorithm}[!t]
    \caption{UMWFL with PAM}
        \begin{algorithmic}[1]
            \State  Initialize $\mathbf{F}^{(0)},\,\{r_k^{(0)},t_k^{(0)}\}$ for all $k$.
            \State \textbf{For} $n=0,\cdots,N_{\rm{max}}-1$:
            \State \ \ \textbf{For} $m=0,\cdots,M_{\rm{max}}-1$:
            \State \ \  \ \  Update $\{\mathbf{u}_k^{(m+1)}\}_{k=1}^K$ using \eqref{ukm} with CP $\mathcal{O}(KN^2)$.
            \State \ \   \ \ Update $\mathbf{f}^{(m+1)}$ using \eqref{fm} with CP $\mathcal{O}(KN^2)$.
            \State \ \   \ \ Update $\mathbf{z}^{(m+1)}$ using \eqref{zm} with CP $\mathcal{O}(N^2)$.
            \State \ \ \textbf{End~For}
            \State \ \ Update $\mathbf{F}^{(n+1)}$ as $\mathbf{F}^{(n+1)}=\mathrm{mat}(\mathbf{z}^{(M_{\rm{max}})})$.
            \State \ \ Update $\{r_k^{(n+1)}\}_{k=1}^K$ using \eqref{qk} with CP $\mathcal{O}(KN^2)$.
            \State \ \ Update $\{t_k^{(n+1)}\}_{k=1}^K$ using \eqref{problemp} with CP $\mathcal{O}(K^{3.5})$.
            \State \textbf{End~For}
            \State Output $\mathbf{F}^{\diamond}=\mathbf{F}^{(N_{\rm{max}})}$, $r_k^\diamond=r_k^{(N_{\rm{max}})}$, $t_k^\diamond=t_k^{(N_{\rm{max}})}$.
        \end{algorithmic}
\end{algorithm}

\section{Results and Discussions}

This section presents simulation results to verify the performance of the proposed scheme.
We consider the object detection task in autonomous vehicle (AV) systems \cite{zijian,fedcav2}.
This task involves three main steps: 1) distributed dataset generation and storage at each vehicle; 2) collecting labels from nearby road-side infrastructures (RSIs); 3) federated learning within a vehicle platoon.
The three steps are executed sequentially and may not occur in the same place.
The case of $K=3$ is simulated and the number of total UMWFL iterations is set to $15$.
The average precision at intersection of union (IoU) equal to 0.5 is used for performance evaluation.
Besides the proposed PAM algorithm, we also simulate a baseline scheme, which sets $\mathbf{F}=\mathbf{I}_N$ and optimize $\{t_k,r_k\}$ using \eqref{problemq} and \eqref{problemp}.

For dataset generation and storage, it can be realized via a physical-world testbed; but this involves high implementation costs.
Car Learning to Act (CARLA) \cite{carla} is a widely-accepted unreal-engine platform that provides complex urban driving scenarios and high 3D rendering quality such that the AV object detection can be prototyped in virtual-reality.
Hence, in this paper, all the datasets are generated by CARLA.
In particular, we employ CARLA to generate $29$ vehicles in the ``Town05'' map, among which $3$ are autonomous driving vehicles that can generate the point-cloud data at a frequency of $10$ frames/s.
The entire dataset consists of $600$ frames at each vehicle, where $100$ frames are used for training and $500$ frames are used for testing.
Fig.~2a illustrates the bird eye view of the simulated world and the locations of all vehicles.
The datasets at different vehicles are non-i.i.d. due to various fields of views.
For label collection, it is assumed that each vehicle can obtain the ground truth label from its nearby RSI.
The labels should satisfy (Karlsruhe Institute of Technology and Toyota Technological Institute) KITTI formats.
This can be realized by processing the raw data generated from CARLA using the Python scripts in \cite{zijian,shanfeng}.
Finally, to simulate the vehicle platoon federated learning, we need to specify the local model and the communication parameters.
Specifically, the sparsely embedded convolutional detection (SECOND) neural network \cite{second} is adopted for object detection.
Each round of local training processes $100$ frames and the Adam optimizer is adopted with a learning rate of $10^{-4}$.
The number of local updates is $E=1$.
The model training is implemented by PyTorch using Python 3.8 on a Linux server with an NVIDIA RTX 3090 GPU.
For the communication model, the pathloss of the user $k$ is set to $\varrho_{k}=-40\,\mathrm{dB}$, and $\mathbf{h}_{k}$ and $\mathbf{g}_{k}$ are generated according to $\mathcal{CN}(\mathbf{0},\varrho_{k}\mathbf{I}_N)$.
The edge server is equipped with $N=8$ antennas.
The power scaling factor $\gamma=1$ and the maximum transmit powers at users are $P_0=1~\mathrm{W}$ (i.e., $30\,\mathrm{dBm}$).
The noise powers at the server and users are set as $-80\,\mathrm{dBm}$, which capture the effects of thermal noise, receiver noise, and interference.

\textbf{Performance}.
Firstly, it can be seen from Fig.~2b that the training loss of the baseline scheme fluctuates violently at the iterations $7$ and $11$, because the model errors are non-negligible due to the strong multi-path fading at these communication rounds.
In contrast, via proper phase shifting at the server, such a multi-path fading is effectively alleviated by the proposed scheme.
Therefore, the training loss of the proposed UMWFL with the PAM algorithm converges to a significantly smaller value than that of the baseline scheme.
Secondly, as shown in Fig.~2c, the average precision region of the proposed PAM scheme encloses that of the baseline scheme, which implies that the PAM scheme achieves higher learning performance on all the testing datasets.
Particularly, the performance improvement is between $10\%$ and $50\%$.
Lastly, the qualitative performance of the proposed and baseline schemes are compared in Fig.~3.
It can be observed that the proposed UMWFL with PAM detects all the objects correctly.
However, the baseline scheme misses a turning car at vehicle~$1$, misses an occluded car at vehicle~$2$, and outputs an inaccurate detection at vehicle~$3$.
This demonstrates the non-robustness of wireless FL if no phase optimization is performed.
In contrast, if non-unit-modulus beamforming optimization is employed, significantly higher implementation costs are involved.
The proposed UMWFL can effectively strike a balance between the robustness of the FL and the cost at the server, which is a competitive technology in real edge FL implementations.

\begin{figure*}[!t]
\centering
 \includegraphics[width=0.99\textwidth]{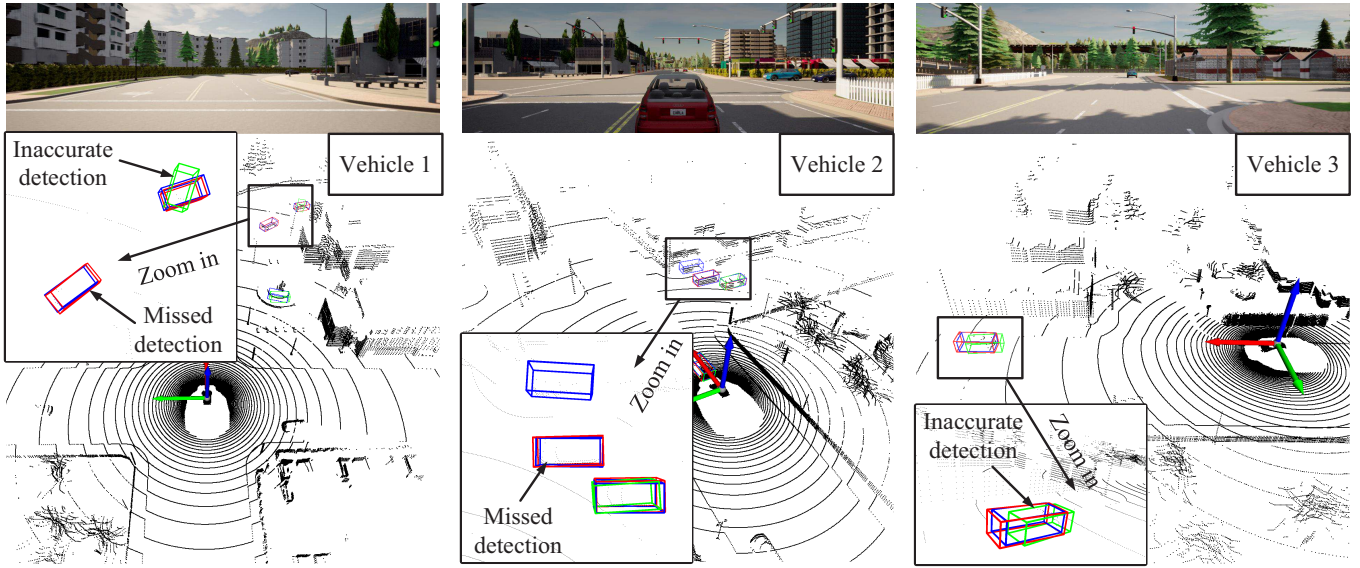}
  \caption{Detection results when $N=8$ with $K=3$. The red box is the ground truth; the blue box is from the proposed UMWFL scheme; the green box is from the fixed beamforming (baseline) scheme. Vehicle 1: The baseline scheme misses a turning car; Vehicle 2: The baseline scheme misses an occluded car; Vehicle 3: The baseline scheme outputs an inaccurate detection.}
\end{figure*}

\section{Conclusion}

This paper proposed the UMWFL scheme to support simultaneous transmission of local model parameters.
The training loss bound was derived and a low-complexity large-scale optimization algorithm was proposed to minimize the training loss.
The performance of the UMWFL framework with the proposed optimization algorithm was verified by using the CARLA autonomous driving platform.
It was found that the UMWFL scheme achieves robust FL performance and low communication costs.

\section{Acknowledgement}

This work was supported in part by the National Natural Science Foundation of China under Grant 62001203, in part by the Shenzhen Science and Technology Program under Grant RCB20200714114956153, and in part by the Shenzhen Fundamental Research Program under Grant JCYJ20190809142403596.

\end{document}